\documentclass[11pt,a4paper]{amsart}
\usepackage[latin1]{inputenc}
\usepackage[all]{xy}
\usepackage{amssymb, amsmath, amscd, geometry}
\usepackage{graphicx}
\usepackage{setspace}

\input xy
\xyoption{all} 

\theoremstyle{plain}
\newtheorem{theorem}{Theorem}[section]
\newtheorem{lemma}[theorem]{Lemma}

\newtheorem{prop}[theorem]{Proposition}

\theoremstyle{definition}

\newcommand{\ppi}{\otimes_\pi}
\newcommand{\pe}{\otimes_\epsilon}

\begin{document}

\title{The communication complexity of XOR games via summing operators}

\author[C. Palazuelos]{Carlos Palazuelos}
\email{cpalazue@illinois.edu}
\address{Department of Mathematics, University of Illinois, Urbana, IL 61801, USA}

\author[D. P\'erez-Garc\'{\i}a]{David P\'erez-Garc\'{\i}a}
\email{dperez@mat.ucm.es}

\address{Departamento de An{\'a}lisis Matem{\'a}tico \\
Facultad de Matem{\'a}ticas\\
\noindent Universidad Complutense de Madrid \\
Madrid 28040\\
Spain}
\author[I. Villanueva]{Ignacio Villanueva}
\email{ignaciov@mat.ucm.es}
\address{Departamento de An{\'a}lisis Matem{\'a}tico \\
Facultad de Matem{\'a}ticas\\
\noindent Universidad Complutense de Madrid \\
Madrid 28040\\
Spain}

\begin{abstract}
The discrepancy method is widely used to find lower bounds for
communication complexity of XOR games. It is well known that these
bounds can be far from optimal. In this context Disjointness is
usually mentioned as a case where the method fails to give good
bounds, because the increment of the value of the game is linear
(rather than exponential) in the number of communicated bits. We
show in this paper the existence of XOR games where the discrepancy
method yields bounds as poor as one desires. Indeed, we show the
existence of such games with any previously prescribed value.
Specifically we prove the  following:

For any number of bits $c$ and every $0<\delta<1$ and for every
$\epsilon>0$, we show the existence of a XOR game such that its
value, both without communication or with the use of $c$ bits of
communication, is contained in the interval $(\delta-\epsilon,
\delta+\epsilon)$.

To prove this result we apply the theory of $p$-summing operators, a
central topic in Banach space theory.  We  show in the paper other
applications of this theory to the study of the communication
complexity of  XOR games.
\end{abstract}

\maketitle

\section{Introduction}
A XOR game $G=(f,\pi)$ with $N$ inputs on each side is defined by a
function $$f:[N]\times [N]:\longrightarrow \{-1,1\}$$ together with
a probability distribution $\pi:[N]\times [N]:\longrightarrow
[0,1]$. Alice and Bob receive as inputs $x,y\in [N]$ respectively
and each of them must answer a number $a,b\in \{-1,1\}$, so that $f(x,y)=a\cdot b$.
They can also be viewed as linear combinations of the correlations
achieved by the two parties when they are asked questions $x,y$.

\smallskip

XOR games are a very natural model for the study of communication
complexity in computation as in Yao's model (\cite{HuNi}). They have
also  been used for the study of complexity classes (\cite{We}),
hardness of approximation (\cite{Ha}), or for a better understanding
of parallel repetition results, both in the classical and the
quantum contexts (\cite{KeRe}, \cite{FK}, \cite{Ra}, \cite{Cl})

In the context of quantum information, XOR games appear often with
the name of correlation Bell inequalities. The so-called CHSH
inequality has an extrordinary relevance in this context
\cite{WeWo}. They also provide an excellent testbed  to study the
relation between classical computation, quantum computation and
communication complexity (\cite{ReTo}, \cite{Bu}).

\smallskip
The discrepancy method is one of the techniques most widely used to
find lower bounds for the communication complexity of a XOR game.
This method is known to give poor bounds in certain cases. We show
in Theorem \ref{Bell} that these bounds can be as poor as one wants.

We introduce the theory of $p$-summing operators with few vectors to
study the communication complexity of XOR games. This leads us to
the use of classical tools in the local theory of Banach spaces,
like Grothendieck inequality, Chevet inequality, $p$-stable measures
and the concentration of measure phenomenon. We must mention here
the papers \cite{LiSh}, \cite{LiNe} where techniques related to ours
have also been used.

\smallskip

We say that a joint strategy between Alice and Bob $\gamma$ is {\em
$c$-simulable} if Alice and Bob can simulate $\gamma$ using $c$-bits
of communication. We denote these strategies by  $\mathcal S_c$. We
say that $\gamma$ is {\em $c$-simulable from Alice to Bob} if they
can simulate it when Alice sends $c$-bits of one way communication
to Bob. We denote these strategies by $\mathcal S^{1}_c$

For a game $G$ we define its values $\omega(G), \omega_c(G),
\omega_c^{1}(G)$ as the maximum value that it attains on the
strategies in $\mathcal L, \mathcal S_c, \mathcal S_c^{1}$
respectively.

\smallskip

The discrepancy method, as stated in \cite[Proposition 3.28]{HuNi}
tells us that, for every game $G$, $\omega_c(G)\leq 2^c \omega(G)$ .
Disjointness function is usually shown as an example where the
discrepancy method fails to give good lower bounds for the
communication complexity, since the increment
$\frac{\omega_c(G)}{\omega(G)}$ is only linear in $c$.

In our main result, we show the existence of XOR games for which the
performance of the discrepancy method is as poor as one  desires.
Specifically, for any prescribed $0<\delta<1$ and for any number of
bits $c$ we prove the existence of a XOR game $G$ such that both
$\omega(G)$ and $\omega_c(G)$ are as close to $\delta$ (and hence
also to each other) as we want.

\begin{theorem}\label{Bell}
For every real number $0<\delta<1$, for every $c\in \mathbb N$ and
for every  $\epsilon>0$ there exists a natural number $N$ and a XOR
game $G$ with $N$ inputs per player such that:
$$\omega(G) \sim_\epsilon  \delta \sim_\epsilon \omega_c(G),$$
where we use the notation $a\sim_\epsilon b$ to denote
$b-\epsilon\leq a\leq b+\epsilon$ for every $a,b,\epsilon > 0$.
\end{theorem}

\

Next, we study how sharp the bound given by the discrepancy method
is, taking also into account the number of inputs $N$. We do a full
study for the case of one-way communication. Our techniques can also
be applied to more general cases. We use the notation $\simeq$ to
denote equality up to universal constants (independent of $N$ and
$c$).

\begin{theorem}\label{distancias}
For every XOR game $G$ with $N$ inputs per player, we have:

\begin{enumerate}
\item[a)] $\omega_c^{1}(G)\leq K_G 2^\frac{c}{2}\omega(G)$ and \item[b)] $\omega_c^{1}(G)\geq \frac{2^\frac{c}{2}}{K_G \sqrt{N}} $.

These inequalities are tight in the sense that there exist games
$J,H$ with $N$ inputs per player such that such that
\begin{equation}\label{tight1}
\omega_c^{1}(J)\simeq 2^\frac{c}{2}\omega(J)
\mbox{ and }
\end{equation}
\begin{equation}\label{tight2}
\omega_c^{1}(H)\simeq \frac{2^\frac{c}{2}}{\sqrt{N}}.
\end{equation}
\end{enumerate}

\end{theorem}

Actually, in Proposition \ref{concentracion} below, we prove that, for
big values of $N$,``most'' games verify conditions (\ref{tight1}) and
(\ref{tight2}).

\

The structure of the paper is the following: In Section
\ref{notacion} we introduce the notation and the formalism we will
use. Next we introduce the  mathematical tools that we need
($p$-summing operators with few vectors, Grothendicek inequality and
Chevet inequality) and finally we state and prove the link between
communication complexity and the theory of $p$-summing operators
with few vectors.

Section \ref{prueba2} is devoted to the proof of Theorem \ref{Bell},
and Theorem \ref{distancias} is proved in Section \ref{prueba1}.

\section{Notation and mathematical tools}\label{notacion}

\subsection{Notation}

We will need the following notations and results from tensor product
theory. Given a normed space $X$, we write $X^*$ for its dual space
with its natural dual norm and $B_{X}$ for its closed unit ball.
Given two normed spaces $X,Y$, and an element $u\in X\otimes Y$, we
define its {\em projective} norm $\|u\|_\pi$ as
$$\|u\|_\pi=\inf \{\sum_i \|x^i\| \|y^i\|, \mbox{ where } u=\sum_i
x^i\otimes y^i\}$$ We write $X\ppi Y$ for the tensor product of $X$
and $Y$ endowed with the projective norm.

We can also define the {\em injective} norm of $u$ as
$$\|u\|_\epsilon=\sup \{\sum_i x^*(x^i) y^*(y^i), \mbox{ where } u=\sum_i
x^i\otimes y^i, \, x^*\in B_{X^*}, y^*\in B_{Y^*}\}$$ and we write
$X\pe Y$ for their tensor product endowed with the injective norm.

In this note, $X$ and $Y$ will always be finite dimensional. It is
well known (and not hard to see) that in that case $(X\ppi
Y)^*=X^*\pe Y^*$ and $(X \pe Y)^*=X^*\ppi Y^*$.

\smallskip

In this paper we will always see games $(G_{x,y})_{x,y=1}^N=G$ as elements in $(\ell_\infty^N\otimes \ell_\infty^N)^*$, the algebraic dual of
$\ell_\infty^N\otimes \ell_\infty^N$. We view the correlations attained by the players (or strategies) as elements $(\gamma_{x,y})_{x,y=1}^N=\gamma$ in $\ell_\infty^N\otimes \ell_\infty^N$. The value of the game $G$ when the players play the strategy $\gamma$ is
$$\langle G,\gamma\rangle=\sum_{x,y=1}^N G_{x,y}
\gamma_{x,y}.$$

For an element $G\in (\ell_\infty^N\otimes \ell_\infty^N)^*$, we write
$\|G\|_{op}$ for its norm as an element of $(\ell_\infty^N\ppi
\ell_\infty^N)^*$, which coincides with its operator norm when we
identify $G$ with the  operator
$\tilde{G}:\ell_\infty^N\longrightarrow \ell_1^N=(\ell_\infty^N)^*$
defined by $\tilde{G}(x)(y)=G(x,y)$. The way they are defined, XOR
games are normalized in the sense that their norm as elements of
$(\ell_\infty^N\pe \ell_\infty^N)^*$ is always one.

\smallskip

When they do not communicate, Alice strategy upon receiving input
$x$ can be described as an element $\alpha(x,\lambda)$ in
$\ell_\infty^{N}$, where $\lambda$ stands for the state of their
shared randomness. Similarly, Bob's strategy is $\beta(y,\lambda)$
so that their joint strategy is an element
$(\gamma_{x,y})_{x,y=1}^N=\gamma=\sum_i \lambda_i \alpha^i\otimes
\beta^i\in \ell_\infty^N\otimes\ell_\infty^N$ such that
$\|\gamma\|_\pi\leq 1$. We will call these strategies {\em local}
strategies and denote it by $\mathcal L$.

\subsection{Summing operators}
Following Grothendieck's work \cite{resume}, the so called {\em local theory of Banach spaces}
has been one of the cornerstones of modern functional analysis. Many
of the main results in this theory can be expressed in terms of {\em
summing operators}. We state next the definitions and results that
we use in this paper. A detailed exposition in this area can be
read, for instance, in \cite{DiJaTo}.

Given a  finite sequence with arbitrary length $(x_i)_{i=1}^n$ in a
normed space $X$, and a real number $1\leq p <\infty$, we define the
{\em weakly p-summing} norm of $(x_i)_{i=1}^n$ by
$$\|(x_i)_{i=1}^n\|_p^w=\sup \{\left(\sum_{i=1}^n |x^*(x_i)|^p\right)^\frac{1}{p}, \mbox{ where }
x^*\in B_{X^*}\}.$$

Now, given an operator $T:X\longrightarrow Y$ between normed spaces,
we define its {\em $p$-summing} norm as $$\pi_p(T)=\inf \{C \mbox{
such that } \left(\sum_{i=1}^n \|T(x_i)\|^p\right)^\frac{1}{p}\leq C
\|(x_i)_{i=1}^n\|_p^w \}$$ for every sequence $(x_i)_{i=1}^n\subset
X$. It is well known that, for an operator $T:\ell_\infty^N\longrightarrow Y$, $\pi_1(T)=\sum_{i=1}^N \|T(e_i)\|$.

If we fix $r\in \mathbb N$ and restrict the previous definition to
sequences $(x_i)_{i=1}^r$ of maximum length $r$ we obtain the
definition of the {\em $p$-summing with $r$ vectors} norm of $T$,
which we denote by $\pi_p^r(T)$. Summing operators with few vectors have been studied by several authors,
see for instance  \cite{T-J} and the references therein.

We will also use the following consequence of Grothendieck's
inequality.

\begin{theorem}
There exists an universal constant $K_G$ such that, for any natural numbers $N,M$ and every operator $u:\ell_\infty^N\longrightarrow \ell_1^M$, $$\pi_2(u)\leq K_G \|u\|.$$
\end{theorem}

\subsection{Chevet inequality}
The following result is known as Chevet inequality. It is usually stated for Gaussian random variables, we state it for Bernouilli random variables. See \cite{LedouxTalagrand}.
\begin{theorem}
Given two normed spaces $E, F$, there exists a universal constant $b$ such that

$$\mathbb{E}\|\sum_{x,y} r_{x,y} \varphi_x\otimes \phi_y\|_{E\pe F} \leq b (\|(\varphi_x)_x\|_2^w \mathbb{E} \|\sum_y r_y \phi_y\|_F+ \|(\phi_y)_y\|_2^w \mathbb{E} \|\sum_x r_x \varphi_x\|_E),$$

where $r_x,r_y,r_{x,y}$ are independent Bernouilli random variables, and $(\varphi_x)_x$, $(\phi_y)_y$ are finite sequences in $E,F$. Actually, we can take $b=\sqrt{\frac{\pi}{2}}$ (see \cite{Defant}).
\end{theorem}

We will apply Chevet inequality in the case $E=F=\ell_1^N$, and $\varphi_x=e_x$, $\phi_y=e_y$, $1\leq x,y \leq N$. In that case, we get easily $$\|(e_x)_{x=1}^N\|_2^w=\sqrt{N} \mbox{ and } \mathbb{E} \|\sum_{y=1}^N r_y e_y\|_{\ell_1^N}=N$$

\subsection{Commnunication complexity and $p$-summing operators}

In this paper, we approach the study of the different types of
strategies and the corresponding value of the games through tensor
norms in  $\ell_\infty^N\otimes \ell_\infty^N$ and its dual space.
We have already mentioned that the local strategies can be
identified with the norm unit ball of $\ell_\infty^N\ppi
\ell_\infty^N$. It is easy to see that both $\mathcal S^{1}_c$ and
$\mathcal S_c$ are symmetric convex bodies of $\ell_\infty^N\otimes
\ell_\infty^N$ with non empty interior. Hence, they define norms on
this space, and therefore the value of a game $G$ on them can be
seen as the corresponding dual norm of the game.

Lemma \ref{startingpoint} below is the starting point of our approach: It
identifies the value of a game on the strategies in $\mathcal
S^{1}_c$ as certain operator norm. We isolate the technical parts of the proof in the following lemma. Its proof follows from \cite[Proposition 2.2 and Lemma 16.13]{DiJaTo}.

\begin{lemma}\label{tecnico1}
Let $B=\{ (\alpha_i)_{i=1}^r\subset \ell_\infty^N \mbox{ such that } \|(\alpha_i)_{i=1}^r\|_1^w \leq 1\}$. Given $A\subset [N]$, let $\alpha_A\in \ell_\infty^N$ be the element defined by $\alpha_A(x)=1$ if $i\in A$ and $\alpha_A(x)=0$ otherwise. Let  $A=\{ (\alpha_{A_i})_{i=1}^r; \mbox{ where } A_1,\ldots, A_r \mbox{ is a partition of } [N]\}$. Then $B$ is the symmetric convex hull of the elements in $A$.
\end{lemma}

\begin{lemma}\label{startingpoint} Given a game  $(G_{x,y})_{x,y=1}^N$, $\omega(G)=\|\tilde{G}\|_{op}$ and $\omega_c^{1}(G)=\pi_1^{2^c}(\tilde{G})$.
\end{lemma}

\begin{proof}
The first statement follows immediately by duality from the characterization of the local strategies.

We prove the second statement. Let us first see that
$\omega^{1}_c(G)\leq \pi_1^{2^c}(\tilde{G})$. We assume that the
communication Alice sends might be dependent on a variable
$\lambda\in \Lambda$. We call $T(x,\lambda)$ to the word that Alice
sends when she receives the input $x$ and the random variable takes
the value $\lambda$. We have that $1\leq T(x,\lambda)\leq 2^c$. For
fixed $1\leq i \leq 2^c$ and $1\leq x \leq N$, call
$$\Lambda_{i,x}=\{\lambda\in \Lambda \mbox{ such that }
T(x,\lambda)=i\}. $$ For fixed $i, \lambda$, call
$$X_{i,\lambda}=\{x \mbox{ such that } T(x,\lambda)=i\}.$$
Calling $\alpha, \beta$ to the strategies followed by Alice and Bob,
we have

$$\omega^{1}_c(G)=\sup \sum_{x,y}G_{xy}\int_\Lambda \alpha(x,\lambda) \beta(y, \lambda, T(x,\lambda)) d\lambda=$$ $$=\sup \sum_{x,y}G_{xy}\sum_{i}
\int_{\Lambda_{i,x}} \alpha(x,\lambda) \beta(y, \lambda,
T(x,\lambda)) d\lambda=\sum_{x,y}G_{xy}\sum_{i} \int_{\Lambda_{i,x}}
\alpha(x,\lambda) \beta(y, \lambda, i) d\lambda=$$
$$=\sum_{x,y}G_{xy}\sum_{i} \int_{\Lambda_{i,x}} \alpha(x,\lambda)
\beta_i(y, \lambda) d\lambda=\int_{\Lambda} \sum_{i} \sum_{x\in
X_{i,\lambda}}\sum_{y}G_{xy}  \alpha(x,\lambda) \beta_i(y, \lambda)
d\lambda.$$

For fixed $\lambda$, $\sum_{i} \sum_{x\in
X_{i,\lambda}}\sum_{y}G_{xy}  \alpha(x,\lambda) \beta_i(y, \lambda)$
is bounded above by $\pi_1^{2^c}(\tilde{G})$ (use Lemma
\ref{tecnico1} for this). Considering the convex hull will not
change this fact.

\smallskip

The reverse inequality follows easily from Lemma \ref{tecnico1} and convexity.
\end{proof}

We mentioned in the Section 2 that all XOR games with $N$ inputs per
player $G$ have norm one considered as elements of
$(\ell_\infty^N\pe\ell_\infty^N)^*$. It is well known
(\cite{DiJaTo}) that this is equivalent to the condition
$\pi_1(\tilde{G})=\pi_1^N(\tilde{G})=1$. In particular, if $c\geq
\log N$, then $\omega_c^1(G)=\omega_c(G)=1$.

\section{Proof of Theorem \ref{Bell}}\label{prueba2}

Theorem \ref{Bell} follows  from Theorem \ref{main} and Proposition
\ref{cambiodelado} below.

\begin{theorem}\label{main}
For any real number $\alpha> 1$,  positive integer $t$ and $\epsilon> 0$, there exists  a natural number $N$ and an operator $T:\ell_\infty^N\longrightarrow \ell_1^N$ such that

\begin{enumerate}

\item[1)] $\|T\|_{op}\sim_\epsilon 1$

\item[2)] $\pi_1^t(T) \sim_\epsilon 1$

\item[3)] $\pi_1(T) \sim_\epsilon \alpha$
\end{enumerate}
\end{theorem}

The game $G$ that we look for in Theorem \ref{Bell} is nothing but
the game whose associated operator is
$\tilde{G}=\frac{T}{\pi_1(T)}$, with the proper choices of
$\epsilon$, $t$ and $\alpha$.

The key point of the proof of Theorem \ref{main} is Levi's embedding
theorem, which says that, for every $1< p< 2$, we have an isometric
embedding of $\ell_p$ into $L_1[0,1]$. Actually, the result is much
more general (see \cite{MZ} and \cite{Kadec}). This embedding is
highly non-trivial and it is based on  $p$-stable measures. We are
interested in the $(1+\epsilon)$-isomorphic finite dimensional
version of the theorem. Specifically, we use the following
improvement of Levi's embedding theorem due to Johnson and
Schechtman.

\begin{theorem}[Theorem 1, \cite{John-Sche}]\label{Levi's embedding}
Let $\epsilon> 0$, and suppose that $0< r< s< 2$ with $r\leq 1$. Then
there exists $\beta=\beta(\epsilon, r,s)> 0$ so that if $m$ and $n$ are
positive integers with $m\leq \beta n$, then $\ell_s^m$ is
$(1+\epsilon)$-isomorphic to a subspace of $\ell_r^n$.
\end{theorem}

Note that, in the particular case of $r=1$ and $1< p< 2$, Theorem
\ref{Levi's embedding} assures the existence of $\beta=\beta(\tau,
1,p)> 0$ and an isomorphism $A:\ell_p^m\hookrightarrow \ell_1^n$ such
that $$(1-\tau)\|x\|_{\ell_p^m}\leq \|Ax\|_{\ell_1^n}\leq
(1+\tau)\|x\|_{\ell_p^m}$$for every $x\in \ell_p^m$.

\begin{proof}[Proof of Theorem \ref{main}]

Let $\alpha,t$ and $\epsilon$ be as in the statement.
We define:

\begin{enumerate}

\item[] $\theta_0=\log(\alpha)$,
\item[] $m_0=\min\{m\in \mathbb{N}:t^{\frac{\theta_0}{m}}< 1+\epsilon\}$,
\item[] $k=2^{m_0}$ and
\item[] $q=\frac{m_0}{\theta_0}.$

\end{enumerate}

Note that we can assume that $2< q<\infty$. Indeed, if it is not, we
only have to consider a high enough $m_0$. Then, we define $p$ by
$\frac{1}{p}+\frac{1}{q}=1$ (so $1<p<2$). And we will denote $q=p'$.
Note that, $t^{\frac{1}{p'}}< 1+\epsilon$ and
$k^{\frac{1}{p'}}=(2^{m_0})^\frac{\theta_0}{m_0}=2^{\theta_0}=\alpha$.

We begin by considering the operator
$$S:=k^{-\frac{1}{p}}id:\ell_\infty^k\rightarrow \ell_p^k.$$ It is
not difficult to check that $\|S\|=\pi_p(S)=1$ and
$\pi_1(S)=k^{\frac{1}{p'}}=\alpha$ (see for instance \cite{Defant}).

Now, we define the operator $$T:=A\circ S\circ
P:\ell_\infty^N\rightarrow \ell_\infty^k\rightarrow
\ell_p^k\hookrightarrow\ell_1^N,$$ where $N=\frac{k}{\beta}$ for the
$\beta=\beta(\epsilon, 1,p)$ given by Theorem \ref{Levi's
embedding}, $A$ is the associated $1+\epsilon$-isomorphism given by
the same theorem and $P:\ell_\infty^N\rightarrow \ell_\infty^k$
denotes the standard projection. Now, by Theorem \ref{Levi's
embedding} and the injectivity property of the $p$-summing operators
(see for instance \cite{Defant}), we know that $\|T\|\sim_\epsilon
1$, $\pi_p(T)\sim_\epsilon 1$ and $\pi_1(T)\sim_\epsilon \alpha$. We
finish the proof if we show that $\pi_1^t(T)\sim_\epsilon 1$.

To see this, consider a sequence $x_1,\cdots ,x_t\in \ell_\infty ^N$
such that $$\sup\{\sum_{i=1}^t|x^*(x_i)|:x^*\in B_{\ell_1^N}\}\leq
1.$$ Then, $$\sum_{i=1}^t\|T(x_i)\|\leq
t^\frac{1}{p'}(\sum_{i=1}^t\|T(x_i)\|^p)^\frac{1}{p} \leq (1+
\epsilon)^2.$$

A suitable adjust of the $\epsilon$'s finishes the proof.
\end{proof}

This result yields immediately a ``one-way communication'' version
of Theorem \ref{Bell}. For the general version, we need the
following simple result.

\begin{prop}\label{cambiodelado}
Let $G$ be a XOR game and let $c$ be a natural number. Then
$$\omega_c(G)\leq 2^c\omega_c^{1}(G).$$
\end{prop}

\begin{proof}Applying convexity, we know that there exists a partition
$R_1,\ldots ,R_{2^c}$ of $[N]\times [N]$ in rectangles and sign
vectors $(\alpha^i(x))_{x=1}^N,$ $(\beta^i(y))_{y=1}^N$, with $\alpha^i(x)=\pm 1=\beta^i(y)$ for every
$x,y,i$ such that
$$\omega_c(G)= \sum_{i=1}^{2^c} \sum_{x,y\in R_i}  \alpha^i(x)
\beta^i(y) M_{x,y}.$$
For every fixed $1\leq x\leq N$ we define $i(x,y)$ as the unique $i$ such that $(x,y)\in
R_i$ and we consider the
row of signs $(\alpha^{i(x,1)}(x) \beta^{i(x,1)}(1) , \ldots, \alpha^{i(x,N)}(x)
\beta^{i(x,N)}(N))$. It is easy to see that there are at most $2^{2^c}$ different such rows. Clearly, $2^c$ bits suffice Alice to tell Bob which is the row associated to $x$.
\end{proof}

\section{Proof of Theorem \ref{distancias}}\label{prueba1}

\begin{proof}
a) Let $G$ be a XOR game with $N$ inputs per player, and let
$\tilde{G}:\ell_\infty^N\rightarrow \ell_1^N$ be its associated
operator, as in the introduction. Grothendieck's Theorem tells us
that $\pi_2(\tilde{G})\leq K_G \|\tilde{G}\|_{op}$. Now, let
$x_1,\cdots ,x_{2^c}\in \ell_\infty^N$ be a finite sequence such
that $\|(x_i)_{i=1}^{2^c}\|_1^w\leq 1$. Then,
$$\omega_c^{1}(G)=\pi_1^{2^c}(\tilde{G})\leq \sum_{i=1}^{2^c}\|\tilde{G}(x_i)\|\leq
2^{\frac{c}{2}}(\sum_{i=1}^{2^c}\|\tilde{G}(x_i)\|^2)^{\frac{1}{2}}\leq
2^{\frac{c}{2}} K_G \|\tilde{G}\|_{op}=2^{\frac{c}{2}} K_G
\omega(G).$$

Let us see the optimality. Recall that we view games as elements in
$(\ell_\infty^N\otimes \ell_\infty^N)^*=\ell_1^N\otimes \ell_1^N$.
We apply Chevet inequality, to find a choice of signs
$(\varepsilon_{x,y})_{x,y=1}^N$such that
$$\|\sum_{x,y=1}^{2^c}\varepsilon_{x,y}e_x\otimes
e_y\|_{\ell_1^{2^c}\otimes_\epsilon \ell_1^{2^c}}\leq 1  \mbox{ and
}$$ $$\|\sum_{x,y=1}^{2^c}\varepsilon_{x,y}e_x\otimes
e_y\|_{\ell_1^{2^c}\otimes_\pi \ell_1^{2^c}}\succeq \sqrt{2^c},$$
where $\succeq$ denotes inequality up to an universal constant. This
defines an operator $T:\ell_\infty^{2^c}\rightarrow \ell_1^{2^c}$
such that $\|T\|_{op}\leq 1$ and $\pi_1(T)=\pi_1^{2^c}(T)\succeq
\sqrt{{2^c}}$. We define $T'=\frac{T}{\pi_1(T)}$. Let now
$P:\ell_\infty^N\longrightarrow \ell_\infty^{2^c}$ be the canonical
projection onto the first ${2^c}$ coordinates, and let
$\varphi:\ell_1^{2^c}\longrightarrow \ell_1^N$ be the canonical
inclusion into the first ${2^c}$ coordinates. Then the game $J$
defined by  $\tilde{J}:\varphi\circ T'\circ
P:\ell_\infty^N\rightarrow \ell_\infty^{2^c}\rightarrow
\ell_1^{2^c}\rightarrow \ell_1^N$ verifies what we wanted.

b)Let $G$ be as in the hypothesiss. First we assume that
$\frac{N}{2^c}=h\in \mathbb{N}$. Call $A_j$ to the isometric copy of
$\ell_\infty^{\frac{N}{2^c}}$ contained naturally in $\ell_\infty^N$
considering only the basis elements $e_i$, with
$(j-1)\frac{N}{2^c}<i\leq j \frac{N}{2^c}$. Then
$$1=\pi_1(\tilde{G})=\sum_{i=1}^N \|\tilde{G}(e_i)\|=\sum_{j=1}^{2^c}
\sum_{i=1}^{\frac{N}{2^c}} \|\tilde{G}(e_{(j-1)\frac{N}{2^c}  +
i})\|\leq$$ $$\leq \sum_{j=1}^{2^c} K_G \sqrt{\frac{N}{2^c}}
\|\tilde{G}_{|_{A_i}}\|_{op} \leq K_G \sqrt{\frac{N}{2^c}}
\pi_1^{2^c}(\tilde{G})=K_G \sqrt{\frac{N}{2^c}} \omega_c^{1}(G).$$



Now we consider the case when $\frac{N}{2^c}$ is not an integer, and
we denote $p$ the smallest natural number such that
$\frac{N}{2^c}\leq p$. Then again we have
$\pi_1(G)=\pi_1^N(G)=\pi_1^{p2^c}(G)\leq \sqrt{p}K_G\pi_1^{2^c}
(G)\leq 2K_G\sqrt{\frac{N}{2^c}}\pi_1^{2^c}(G)$ and the result
follows.

We see now the optimality of this result. Apply again Chevet
inequality to find a choice of signs $(\varepsilon_{x,y})_{x,y=1}^N$
such that $\|\sum_{x,y=1}^N\varepsilon_{x,y}e_x\otimes
e_y\|_{\ell_1^N\otimes_\epsilon \ell_1^N}\leq 1$ and
$\|\sum_{x,y=1}^N\varepsilon_{x,y}e_x\otimes
e_y\|_{\ell_1^N\otimes_\pi \ell_1^N}\succeq \sqrt{N}$. Let
$G':\ell_\infty^N\longrightarrow \ell_1^N$ be its associated
operator and let $G$ be the game associated to
$\frac{G'}{\pi_1(G')}$. By a), we know that $\omega_c^{1}(G)\preceq
\frac{2^\frac{c}{2}}{\sqrt{N}} $.
\end{proof}

Actually, we can see that, for big values of $N$,``most'' games
essentially attain the bounds given above. We write the statement
for the case of  games $G=(f,u)$ with $u$ the uniform distribution.
Similar results can be proved for other distributions. The tool now
is the Concentration of Measure Phenomenon.

\begin{prop}\label{concentracion}
Let $X_N$ be the set of games with $N$ inputs per player defined by
$G=(f,u)$, with $u$ the uniform distribution. Consider in $X_N$ the
probability $\mu:\mathcal P(X_N)\longrightarrow [0,1]$ defined by
$\mu(A)=\frac{Card(A)}{2^{N^2}}$. Let $r>0$. If $m$ is a median of
$\omega(G)$ under $\mu$, then $$\mu(\{G \mbox{ such that }
|\omega(G)-m| \geq r\})\leq 2e^{-\frac{N^2 r^2}{2}}.$$
\end{prop}
\begin{proof}
The proof follows immediately from \cite[Proposition 1.3]{Ledoux}
once we check that $\omega(G)$ is a 2-Lipschitz function under the
normalized Hamming distance in $X_N$.
\end{proof}


\begin{thebibliography}{9}
\bibitem{Bu} H. Buhrman, R. Cleve, S. Massar, R. de Wolf, Non-locality and Communication Complexity, arXiv:0907.3584v1 [quant-ph].

\bibitem{Cl} R. Cleve, W. Slofstra, F. Unger, S. Upadhyay, Strong Parallel Repetition Theorem for Quantum XOR Proof Systems, {\em Computational Complexity}, {\bf 17} (2), 282-299 (2008).


\bibitem{Defant} A. Defant and K. Floret, \emph{Tensor Norms and Operator Ideals}, North-Holland, (1993).

\bibitem{Degorre} J.Degorre, M. Kaplan, S. Laplante, J. Roland, The communication complexity of non-signaling distributions, {\em International Symposium on Mathematical Foundations of Computer Science (MFCS'09)}, volume 5734 of Lecture Notes in Computer Science, pages 270-281. Springer, (2009)

\bibitem{DiJaTo} J. Diestel, H. Jarchow, A. Tonge, \emph{Absolutely Summing Operators}, Cambridge University Press, Cambridge (1995).

\bibitem{FK} U. Feige, G. Kindler, R. O'Donnell, Understanding Parallel Repetition Requires Understanding Foams. {\em CCC 2007}, 179-192.

\bibitem{resume} A. Grothendieck, R\'esum\'e de la th\'eorie m\'etrique des produits tensoriels topologiques, {\em Boll. Soc. Mat. S$\tilde{\mbox{a}}$o-Paulo,} {\bf 8} , 1-79 (1956).


\bibitem{Ha} J. Hastad,  Some optimal inapproximability results, {\em J. ACM}, {\bf 48} (4), 798-859 (2001).

\bibitem{HuNi} E. Hushilevitz, N. Nisan, {\em Communication Complexity}, Cambridge University Press (1997).

\bibitem{John-Sche} W. B. Johnson, G. Schechtman, Embedding $\ell_p^m$ into $\ell_1^n$, {\em Acta Math.}, {\bf 149}, 1-2, 71-85 (1982).


\bibitem{Kadec} M. J. Kadec, On linear dimension of spaces $L_p$ and $\ell_q$, {\em Uspechi Mat. Nauk},  \textbf{84}, 95-98 (1958).


\bibitem{KeRe} J. Kempe, O. Regev, No Strong Parallel Repetition with Entangled and Non-signaling Provers, arXiv:0911.0201v1 [quant-ph].

\bibitem{Ledoux} M. Ledoux, {\em The Concentration of Measure Phenomenon}, American Mathematical Society, (2001).

\bibitem{LedouxTalagrand} M. Ledoux, M.Talagrand, \emph{Probability in {B}anach {S}paces},
  Springer-Verlag, (1991).


\bibitem{LiNe} N. Linial, S. Mendelson, G. Schechtman and A. Shraibman, Complexity Measures of Sign Matrices, {\em Combinatorica}, {\bf 27}(4), 439-463 (2007).

\bibitem{LiSh} N. Linial, A. Shraibman, Lower Bounds in Communication Complexity Based on Factorization Norms, {\em Random Structures and Algorithms}, {\bf 34}, 368-394 (2009).

\bibitem{MZ} J. Marcinkiewicz, A. Zigmund, Quelques inequalites pour les operations lineaires, {\em Fund. Math.}, \textbf{32}, 115-121 (1939).


\bibitem{Ra} R. Raz, A Counterexample to Strong Parallel Repetition, {\em Proceeding of the 49th FOCS} (2008).

\bibitem{ReTo} O. Regev, B. Toner, Simulating Quantum Correlations with Finite Communication, {\em Proceedings of 48th Annual IEEE Symposium on Foundations of Computer Science (FOCS 2007)}.


\bibitem{T-J} N. Tomczak-Jaegermann, \emph{Banach-Mazur distances and finite-dimensional operator ideals}, Longmann, (1988).

\bibitem{We} S. Wehner, Entanglement in Interactive Proof Systems with Binary Answers, STACS 2006, 162--171.


\bibitem{WeWo} R. F. Werner, M. M. Wolf, Bell inequalities and Entanglement, {\em Quant. Inf. Comp.}, {\bf 1} (3)  (2001).


\end{thebibliography}
\end{document}